\documentclass[12pt,reqno]{article}
\usepackage[usenames]{color}
\usepackage{amssymb}
\usepackage{graphicx}
\usepackage{amscd}

\usepackage{graphics,amsmath,amssymb,relsize}
\usepackage{bm}
\usepackage{amsthm}
\usepackage{amsfonts}
\usepackage{latexsym}
\usepackage{epsf}

\newtheorem{theorem}{Theorem}[section]
\newtheorem{corollary}[theorem]{Corollary}
\newtheorem{lemma}[theorem]{Lemma}
\newtheorem{proposition}[theorem]{Proposition}

\newtheorem{defin}[theorem]{Definition}
\newenvironment{definition}{\begin{defin}\normalfont\quad}{\end{defin}}
\newtheorem{examp}[theorem]{Example}

\newtheorem{rema}[theorem]{Remark}
\newtheorem{prob}[theorem]{Problem}
\newenvironment{remark}{\begin{rema}\normalfont}{\end{rema}}

\numberwithin{equation}{section}

\newcommand{\bt}{\begin{thm}}
\newcommand{\et}{\end{thm}}
\newcommand{\bp}{\begin{proof}}
\newcommand{\ep}{\end{proof}}
\newcommand{\bprop}{\begin{prop}}
\newcommand{\eprop}{\end{prop}}
\newcommand{\bl}{\begin{lemma}}
\newcommand{\el}{\end{lemma}}
\newcommand{\bc}{\begin{corollary}}
\newcommand{\ec}{\end{corollary}}
\newcommand{\Z}{\mathbb{Z}}

\newcommand{\be}{\begin{enumerate}}
\newcommand{\ee}{\end{enumerate}}

\newcommand{\K}{\mathcal{K}}

\newcommand{\Enc}{\mathcal{E}}
\newcommand{\Dec}{\mathcal{D}}
\newcommand{\Tag}{\mathcal{T}}
\newcommand{\Mac}{\mathcal{M}}
\newcommand{\Auth}{\mathcal{V}}
\newcommand{\So}{\mathcal{S}}
\newcommand{\Me}{\mathcal{M}}
\newcommand{\OMIT}[1]{}

\title{On an almost-universal hash function family \\ with applications to authentication and secrecy codes\footnote{We have presented an extended abstract of this paper \cite{BKSTT21} in ISITA 2016.}}
\author{Khodakhast Bibak \thanks{Department of Computer Science, University of Victoria, Victoria, BC, Canada V8W 3P6. Email: {\tt \{kbibak,bmkapron,srinivas\}@uvic.ca}} \and Bruce M. Kapron \footnotemark[2] \and Venkatesh Srinivasan \footnotemark[2] \thanks{Centre for Quantum Technologies, National University of Singapore, Singapore 117543.} \and L\'aszl\'o T\'oth \thanks{Department of Mathematics, University of P\'ecs, 7624 P\'ecs, Hungary. Email: {\tt ltoth@gamma.ttk.pte.hu}}}

\begin{document}

\maketitle

\begin{abstract}
Universal hashing, discovered by Carter and Wegman in 1979, has many important applications in computer science. MMH$^*$, which was shown to be $\Delta$-universal by Halevi and Krawczyk in 1997, is a well-known universal hash function family. We introduce a variant of MMH$^*$, that we call GRDH, where we use an arbitrary integer $n>1$ instead of prime $p$ and let the keys $\mathbf{x}=\langle x_1, \ldots, x_k \rangle \in \mathbb{Z}_n^k$ satisfy the conditions $\gcd(x_i,n)=t_i$ ($1\leq i\leq k$), where $t_1,\ldots,t_k$ are given positive divisors of $n$. Then via connecting the universal hashing problem to the number of solutions of restricted linear congruences, we prove that the family GRDH is an $\varepsilon$-almost-$\Delta$-universal family of hash functions for some $\varepsilon<1$ if and only if $n$ is odd and $\gcd(x_i,n)=t_i=1$ $(1\leq i\leq k)$. Furthermore, if these conditions are satisfied then GRDH is $\frac{1}{p-1}$-almost-$\Delta$-universal, where $p$ is the smallest prime divisor of $n$. Finally, as an application of our results, we propose an authentication code with secrecy scheme which strongly generalizes the scheme studied by Alomair et al. [{\it J. Math. Cryptol.} {\bf 4} (2010), 121--148], and [{\it J.UCS} {\bf 15} (2009), 2937--2956].
\end{abstract}


\section{Introduction}\label{Sec 1}

Universal hash functions, discovered by Carter and Wegman \cite{CW}, have many applications in computer science, including cryptography and information security \cite{BHKKR, DORS, HK, HP, HAYA1, HAYA2, RW, TYVA, WC}, pseudorandomness \cite{HILL, NIS}, complexity theory \cite{RUWI, SIP}, randomized algorithms \cite{IMZU, MORA}, data structures \cite{PAPA, SIE}, and parallel computing \cite{KSV, LSS}. Since universality of hash functions and its variants are concepts central to this work, we begin by describing them in detail. Our description of these concepts closely follows the definitions given in \cite{HK}.

\subsection{Universal hashing and its variants}\label{Sec 2}

Let $D$ and $R$ be finite sets. Let $H$ be a family of functions from domain $D$ to range $R$. We say that $H$ is a {\it universal} family of hash functions (\cite{CW}) if the probability, over a random choice of a hash function from $H$, that two distinct keys in $D$ have the same hash value is at most $1/|R|$. That is, universal hashing captures the important property that distinct keys in $D$ do not {\it collide} too often. Furthermore, we say that $H$ is an {\it $\varepsilon$-almost-universal} ($\varepsilon$-AU) family of hash functions if the probability of collision is at most $\varepsilon$, for $\frac{1}{|R|} \leq \varepsilon < 1$. In other words, an $\varepsilon$-AU family, for sufficiently small $\varepsilon$, is {\it close} to being universal; see Definition~\ref{def:Uni hash} below. Universal and almost-universal hash functions have many applications in algorithm design. For example, they have been used to provide efficient solutions for the dictionary problem in which the goal is to maintain a dynamic set that is updated using insert and delete operations using small space so that membership queries that ask if a certain element is in $S$ can be answered quickly.

Motivated by applications to cryptography, a notion of $\Delta$-universality was introduced in \cite{K, R, S}. Suppose that $R$ is an Abelian group. We say that $H$ is a $\Delta$-{\it universal} family of hash functions if the probability, over a random $h \in H$, that two distinct keys in $D$ hash to values that are distance $b$ apart for any $b$ in $R$ is 1/$|R|$. Note that the case $b=0$ corresponds to universality. Furthermore, we say that $H$ is {\it $\varepsilon$-almost-$\Delta$-universal} ($\varepsilon$-A$\Delta$U) if this probability is at most $\varepsilon$, $\frac{1}{|R|} \leq \varepsilon < 1$. We remark that $\varepsilon$-A$\Delta$U families have applications to message authentication. Informally, it is possible to design a message authentication scheme using 
$\varepsilon$-A$\Delta$U families such that two parties can exchange signed messages over an unreliable channel and the probability that an adversary can forge a valid signed message to be sent across the channel is at most $\varepsilon$ (\cite{HK}). Also, the well-known leftover hash lemma states that (almost) universal hash functions are good randomness extractors.

Finally, in Section~\ref{Sec_4} on authentication codes with secrecy, we need the notion of strong universality which was introduced in \cite{WC}. We say that $H$ is a {\it strongly universal} family of hash functions if the probability, over a random choice of a hash function from $H$, that two distinct keys $x$ and $y$ in $D$ are mapped to $a$ and $b$ respectively is $1/|R|^2$. We say that $H$ is {\it $\varepsilon$-almost-strongly-universal} ($\varepsilon$-ASU) if this probability is at most $\varepsilon$, $\frac{1}{|R|^2} \leq \varepsilon < \frac{1}{|R|}$.

We now provide a formal definition of the concepts introduced above as in \cite{HK}. For a set 
$\mathcal{X}$, we write $x \leftarrow \mathcal{X}$ to denote that $x$ is chosen uniformly at random from 
$\mathcal{X}$.

\begin{definition}\label{def:Uni hash}
Let $H$ be a family of functions from a domain $D$ to a range $R$. Let $\varepsilon$ be a constant such that $\frac{1}{|R|} \leq \varepsilon < 1$. The probabilities below, are taken over the random choice of hash function $h$ from the set $H$.

\begin{itemize}
\item
The family $H$ is a {\it universal family of hash functions} if for any two distinct $x,y\in D$, we have $\text{Pr}_{h \leftarrow H}[h(x) = h(y)] \leq \frac{1}{|R|}$. Also, $H$ is an 
{\it $\varepsilon$-almost-universal} ($\varepsilon$-AU) {\it family of hash functions} if for any two distinct $x,y\in D$, we have $\text{Pr}_{h \leftarrow H}[h(x) = h(y)] \leq \varepsilon$.

\item
Suppose $R$ is an Abelian group. The family $H$ is a $\Delta$-{\it universal family of hash functions} if for any two distinct $x,y\in D$, and all $b\in R$, we have $\text{Pr}_{h \leftarrow H}[h(x) - h(y) = b] = \frac{1}{|R|}$, where ` $-$ ' denotes the group subtraction operation. Also, $H$ is an {\it $\varepsilon$-almost-$\Delta$-universal} ($\varepsilon$-A$\Delta$U) {\it family of hash functions} if for any two distinct $x,y\in D$, and all $b\in R$, we have $\text{Pr}_{h \leftarrow H}[h(x) - h(y) = b] \leq \varepsilon$.

\item
The family $H$ is a {\it strongly universal family of hash functions} if for any two distinct $x,y\in D$, and all $a,b\in R$, we have $\text{Pr}_{h \leftarrow H}[h(x) = a, \; h(y) = b] = \frac{1}{|R|^2}$. Also, $H$ is an {\it $\varepsilon$-almost-strongly universal} ($\varepsilon$-ASU) {\it family of hash functions} if for any two distinct $x,y\in D$, and all $a,b\in R$, we have $\text{Pr}_{h \leftarrow H}[h(x) = a, \; h(y) = b] \leq \frac{\varepsilon}{|R|}$.
\end{itemize}
\end{definition}

\subsection{MMH$^*$}

The hash function family we study, GRDH, is a variant of a well-known family which was named MMH$^*$ (Multilinear Modular Hashing) by Halevi and Krawczyk \cite{HK}. Let $p$ be a prime and $k$ be a positive integer. Each hash function in the family MMH$^*$ takes as input a $k$-tuple, $\mathbf{m}=\langle m_1, \ldots, m_k \rangle \in \mathbb{Z}_p^k$. It computes the dot product of $\mathbf{m}$ with a fixed $k$-tuple $\mathbf{x}=\langle x_1, \ldots, x_k \rangle \in \mathbb{Z}_p^k$ and outputs this value modulo $p$.

\begin{definition}\label{def:MMH$^*$}
Let $p$ be a prime and $k$ be a positive integer. The family MMH$^*$ is defined as follows:
\begin{align}\label{MMH*}
\text{MMH}^*:=\lbrace g_{\mathbf{x}} \; : \; \mathbb{Z}_p^k \rightarrow \mathbb{Z}_p \; | \; \mathbf{x}\in \mathbb{Z}_p^k \rbrace,
\end{align}
where
\begin{align}\label{MMH* for 2}
g_{\mathbf{x}}(\mathbf{m}) := \mathbf{m} \cdot \mathbf{x} \pmod{p} = \sum_{i=1}^k m_ix_i \pmod{p},
\end{align}
for any $\mathbf{x}=\langle x_1, \ldots, x_k \rangle \in \mathbb{Z}_p^k$, and any $\mathbf{m}=\langle m_1, \ldots, m_k \rangle \in \mathbb{Z}_p^k$.
\end{definition}

The family MMH$^*$ is widely attributed to Carter and Wegman \cite{CW}, while it seems that Gilbert, MacWilliams, and Sloane \cite{GMS} had already discovered it (but in the finite geometry setting). Halevi and Krawczyk \cite{HK}, using the multiplicative inverse method, proved that MMH$^*$ is a 
$\Delta$-universal family of hash functions. We also remark that, recently, Leiserson et al. \cite{LSS} rediscovered MMH$^*$ (called it ``DOTMIX compression function family") and using the same method as of Halevi and Krawczyk \cite{HK} proved that DOTMIX is $\Delta$-universal. Then they apply this result in studying the problem of deterministic parallel random-number generation for dynamic multithreading platforms in parallel computing.

\begin{theorem} {\rm (\cite{HK, LSS})} \label{thm:MMH* UNI}
The family \textnormal{MMH}$^*$ is a $\Delta$-universal family of hash functions.
\end{theorem}

Very recently, it was proved that MMH$^*$ with \textit{arbitrary} modulus is always almost-universal \cite{BKS4}.

\subsection{Our contributions}

Suppose that, instead of a prime $p$, one uses an arbitrary integer $n>1$ in the definition of MMH$^*$. Additionally, we ask that the keys $\mathbf{x}=\langle x_1, \ldots, x_k \rangle \in \mathbb{Z}_n^k$ satisfy the conditions $\gcd(x_i,n)=t_i$ ($1\leq i\leq k$), where $t_1,\ldots,t_k$ are given positive divisors of $n$. We call this new family GRDH and refer the reader to Section~\ref{Sec_3} for a formal definition.

Many natural questions arise: What can we say about universality (or $\varepsilon$-almost-universality) of GRDH? What can we say about $\Delta$-universality (or $\varepsilon$-almost-$\Delta$-universality) of GRDH? Recently, Alomair, Clark, and Poovendran \cite{ACP} presented a construction of codes with secrecy based on a universal hash function family that is a special case of GRDH. Is it possible to generalize their construction and analyse its security properties?

\begin{itemize}
\item
In Theorem~\ref{thm:GRDH e-au}, we prove that if $n,k>1$ then the family GRDH is an $\varepsilon$-AU family of hash functions for some $\varepsilon<1$ if and only if $n$ is odd and $\gcd(x_i,n)=t_i=1$ $(1\leq i\leq k)$. Furthermore, if these conditions are satisfied then GRDH is $\frac{1}{p-1}$-AU, where $p$ is the smallest prime divisor of $n$. This bound is tight.
\item
In Remark~\ref{rem:GRDH one var}, we conclude (from the idea of the proof of Theorem~\ref{thm:GRDH e-au}) that if $k=1$ then the family GRDH is an $\varepsilon$-AU family of hash functions for some 
$\varepsilon<1$ if and only if $\gcd(x_1,n)=t_1=1$. Furthermore, if $\gcd(x_1,n)=t_1=1$ (that is, if $x_1 \in \mathbb{Z}_n^{*}$) then the collision probability for any two distinct messages is `exactly zero'.
\item
In Theorem~\ref{thm:GRDH e-adeltau}, we show that if $n>1$ then the family GRDH is an $\varepsilon$-A$\Delta$U family of hash functions for some $\varepsilon<1$ if and only if $n$ is odd and 
$\gcd(x_i,n)=t_i=1$ $(1\leq i\leq k)$. Furthermore, if these conditions are satisfied then GRDH is 
$\frac{1}{p-1}$-A$\Delta$U, where $p$ is the smallest prime divisor of $n$. This bound is tight.
\item 
In Theorem~\ref{authenc-main}, we generalize the construction of authentication code with secrecy presented in \cite{ACP, AP}. Using Theorem~\ref{thm:GRDH e-adeltau}, we show that our construction is a 
$\frac{1}{(p-1)n^{k-1}},\frac{1}{p-1}$-authentication code with secrecy for equiprobable source states on $\Z_{n}^k \setminus \{\mathbf{0}\}$, where $n$ is odd, and $p$ is the smallest prime divisor of
$n$.
\end{itemize}

Our results show that if one uses a composite integer $n$ in the definition of \textnormal{MMH}$^*$ then even by choosing the keys $\mathbf{x}=\langle x_1, \ldots, x_k\rangle$ from ${\mathbb{Z}_n^{*}}^k$, or more generally, choosing the keys $\mathbf{x}=\langle x_1, \ldots, x_k\rangle$ from $\mathbb{Z}_n^k$ 
with the general conditions $\gcd(x_i,n)=t_i$ ($1\leq i\leq k$), where $t_1,\ldots,t_k$ are given positive divisors of $n$, we cannot get any strong collision bound (unless $k=1$ and $\gcd(x_1,n)=t_1=1$; in this case, as we mentioned above, the collision probability for any two distinct messages is `exactly zero'). Such impossibility results were not known before.

The main technique in proving the hashing results is connecting the universal hashing problem to the number of solutions of restricted linear congruences, which we believe is a novel idea and could be also of independent interest. We use an explicit formula for the number of solutions of restricted linear congruences, recently obtained by Bibak et al. \cite{BKSTT}, using properties of Ramanujan sums and of the finite Fourier transform of arithmetic functions, that we will review in Section~\ref{Sec_2}. We believe that this is the first paper that introduces applications of Ramanujan sums, finite Fourier transform, and restricted linear congruences in the study of universal hashing. We hope this approach will lead to further work.

\section{Restricted linear congruences}\label{Sec_2}

Throughout the paper, we use $(a_1,\ldots,a_k)$ to denote the greatest common divisor (gcd) of the integers $a_1,\ldots,a_k$, and write $\langle a_1,\ldots,a_k\rangle$ for an ordered $k$-tuple of integers. Also, for $a \in \Z \setminus \lbrace 0 \rbrace$, and a prime $p$, we use the notation $p^r\mid\mid a$ if $p^r\mid a$ and $p^{r+1}\nmid a$. We also use $\mathbf{0}$ to denote the vector of all zeroes. The multiplicative group of integers modulo $n$ is denoted by $\mathbb{Z}_n^{*}$.

Let $a_1,\ldots,a_k,b,n\in \Z$, $n\geq 1$. A linear congruence in $k$ unknowns $x_1,\ldots,x_k$ is of the form
\begin{align} \label{cong form}
a_1x_1+\cdots +a_kx_k\equiv b \pmod{n}.
\end{align}
By a solution of (\ref{cong form}), we mean an $\mathbf{x}=\langle x_1,\ldots,x_k \rangle \in \mathbb{Z}_n^k$ that satisfies (\ref{cong form}). The following result, proved by D. N. Lehmer \cite{LEH2}, gives the number of solutions of the above linear congruence:

\begin{proposition}\label{Prop: lin cong}
Let $a_1,\ldots,a_k,b,n\in \Z$, $n\geq 1$. The linear congruence $a_1x_1+\cdots +a_kx_k\equiv b \pmod{n}$ has a solution $\langle x_1,\ldots,x_k \rangle \in \Z_{n}^k$ if and only if $\ell \mid b$, where
$\ell=(a_1, \ldots, a_k, n)$. Furthermore, if this condition is satisfied, then there are $\ell n^{k-1}$ solutions.
\end{proposition}

The solutions of the above congruence may be subject to certain conditions, such as $(x_i,n)=t_i$ ($1\leq i\leq k$), where $t_1,\ldots,t_k$ are given positive divisors of $n$. The number of solutions of this kind of congruence, which were called {\it restricted linear congruences} in \cite{BKSTT}, have been studied, in special cases, in many papers and have found very interesting applications in number theory, combinatorics, and cryptography, among other areas (see \cite{BKSTT3, COH0, DIX, JAWILL, LIS, MENE, NV, San2009, SanSan2013, SY2014, TOT}). Recently, Bibak et al. \cite{BKSTT} dealt with the problem in its `most general case' and using properties of Ramanujan sums and of the finite Fourier transform of arithmetic functions gave an explicit formula for the number of solutions of the restricted linear congruence
\begin{equation} \label{gen_rest_cong}
a_1x_1+\cdots +a_kx_k\equiv b \pmod{n}, \quad (x_i,n)=t_i \ (1\leq i\leq k),
\end{equation}
where $a_1,t_1,\ldots,a_k,t_k, b,n$ ($n\geq 1$) are arbitrary integers.

The special case of $k=2$, $a_i=1$, $t_i=1$ ($1\leq i\leq k$) of \eqref{gen_rest_cong} is related to a long-standing conjecture due to D. H. Lehmer from 1932. Also, the special case of $b=0$, $a_i=1$, $t_i=\frac{n}{m_i}$, $m_i\mid n$ ($1\leq i\leq k$) is related to the {\it orbicyclic} (multivariate arithmetic) function (\cite{LIS}), which has very interesting combinatorial and topological applications, in particular, in counting non-isomorphic maps on orientable surfaces. See \cite{BKSTT} for a detailed discussion about restricted linear congruences and their applications.

If in \eqref{gen_rest_cong} one has $a_i=0$ for every $1\leq i\leq k$, then clearly there are solutions 
$\langle x_1,\ldots,x_k\rangle$ if and only if $b\equiv 0\pmod{n}$ and $t_i \mid n$ ($1\leq i\leq k$), and in this case there are $\varphi(n/t_1)\cdots \varphi(n/t_k)$ solutions.

Consider the restricted linear congruence \eqref{gen_rest_cong} and assume that there is an $i_0$ such that $a_{i_0}\ne 0$. For every prime divisor $p$ of $n$ let $r_p$ be the exponent of $p$ in the prime factorization of $n$ and let $\mathfrak{m}_p=\mathfrak{m}_p(a_1,t_1,\ldots,a_k,t_k)$ denote the smallest $j\geq 1$ such that there is some $i$ with $p^j \nmid a_it_i$. There exists a finite $\mathfrak{m}_p$ for every $p$, since for a sufficiently large $j$ one has $p^j\nmid a_{i_0}t_{i_0}$. Furthermore, let
$$
e_p = e_p(a_1,t_1,\ldots,a_k,t_k) = \# \{i: 1\leq i\leq k, p^{\mathfrak{m}_p}\nmid a_it_i \}.
$$
By definition, $1 \leq e_p \leq$ the number of $i$ such that $a_i\ne 0$. Note that in many situations instead of $\mathfrak{m}_p(a_1,t_1,\ldots,a_k,t_k)$ we write $\mathfrak{m}_p$ and instead of $e_p(a_1,t_1,\ldots,a_k,t_k)$ we write $e_p$ for short. However, it is important to note that both $\mathfrak{m}_p$ 
and $e_p$ always depend on $a_1,t_1,\ldots,a_k,t_k,p$.

\begin{theorem} {\rm (\cite{BKSTT})} \label{th_gen_expl} Let $a_i,t_i, b,n\in \Z$, $n\geq 1$, $t_i\mid n$ {\rm ($1\leq i\leq k$)} and assume that $a_i\neq 0$ for at least one $i$. Consider the linear congruence $a_1x_1+\cdots +a_kx_k\equiv b \pmod{n}$, with $(x_i,n)=t_i$ {\rm ($1\leq i\leq k$)}. If there is a prime $p\mid n$ such that $\mathfrak{m}_p\leq r_p$ and $p^{\mathfrak{m}_p-1}\nmid b$ or $\mathfrak{m}_p\geq r_p+1$ and $p^{r_p}\nmid b$, then the linear congruence has no solution. Otherwise, the number of solutions is
\begin{equation} \label{main_prod_formula}
\mathlarger{\prod}_{i=1}^{k} \varphi\left(\frac{n}{t_i}\right)
\mathlarger{\prod}_{\substack{p\,\mid\, n \\ \mathfrak{m}_p \,\leq \, r_p \\ p^{\mathfrak{m}_p} \,\mid\, b}} p^{\mathfrak{m}_p-r_p-1} \left(1-\frac{(-1)^{e_p-1}}{(p-1)^{e_p-1}} \right)
\mathlarger{\prod}_{\substack{p\, \mid\, n \\ \mathfrak{m}_p \,\leq \, r_p \\ p^{\mathfrak{m}_p-1} \, \|\, b}} p^{\mathfrak{m}_p-r_p-1} \left(1-\frac{(-1)^{e_p}}{(p-1)^{e_p}}\right),
\end{equation}
where the last two products are over the prime factors $p$ of $n$ with the given additional properties. Note that the last product is empty and equal to $1$ if $b=0$.
\end{theorem}

Formula (\ref{main_prod_formula}) will be the core for the applications to universal hashing that we present in this paper.

\begin{corollary} {\rm (\cite{BKSTT})} \label{cor_zero} The restricted congruence given in Theorem~\ref{th_gen_expl} has no solutions if and only if one of the following cases holds:

(i) there is a prime $p\mid n$ with $\mathfrak{m}_p\leq r_p$ and $p^{\mathfrak{m}_p-1}\nmid b$;

(ii) there is a prime $p\mid n$ with $\mathfrak{m}_p\geq r_p+1$ and $p^{r_p}\nmid b$;

(iii) there is a prime $p\mid n$ with $\mathfrak{m}_p\leq r_p$, $e_p=1$ and $p^{\mathfrak{m}_p}\mid b$;

(iv) $n$ is even, $\mathfrak{m}_2 \leq r_2$, $e_2$ is odd and $2^{\mathfrak{m}_2}\mid b$;

(v) $n$ is even, $\mathfrak{m}_2 \leq r_2$, $e_2$ is even and $2^{\mathfrak{m}_2-1}\, \|\, b$.
\end{corollary}

Corollary~\ref{cor_zero} is the only result in the literature which gives \textit{necessary and sufficient conditions} for the (non-)existence of solutions of restricted linear congruences in their most general case and might lead to interesting applications/implications. For example, Corollary~\ref{cor_zero} can be considered as relevant to the generalized knapsack problem. The {\it knapsack problem} is of significant interest in cryptography, computational complexity, and several other areas. Micciancio \cite{MIC} proposed a generalization of this problem to arbitrary rings, and studied its average-case complexity. This {\it generalized knapsack problem}, proposed by Micciancio \cite{MIC}, is described as follows: for any ring $R$ and subset $S \subset R$, given elements $a_1, \ldots , a_k \in R$ and a target element $b \in R$, find $\langle x_1,\ldots,x_k\rangle \in S^k$ such that $\sum_{i=1}^k a_i \cdot x_i = b$, where all operations are performed in the ring. Interestingly, Corollary~\ref{cor_zero} helps us to deal with this problem in a quite natural case:

\begin{rema} \label{knapsack}
The generalized knapsack problem with $R=\mathbb{Z}_n$ and $S=\mathbb{Z}_n^{*}$ has no solutions if and only if one of the cases of Corollary~\ref{cor_zero} holds.
\end{rema}

Theorem~\ref{th_gen_expl} has also important applications in combinatorics, geometry, string theory, and quantum field theory (QFT) \cite{BKS2}, for example, it is related to the Harvey's famous theorem on the cyclic groups of automorphisms of compact Riemann surfaces \cite{BKS2, LIS}.

\section{GRDH}\label{Sec_3}

In this section, we introduce a variant of MMH$^*$ that we call GRDH (Generalized Restricted Dot Product Hashing). Then we investigate the $\varepsilon$-almost-universality and $\varepsilon$-almost-$\Delta$-universality of GRDH via connecting the problem to the number of solutions of restricted linear congruences.

\begin{definition}\label{def:GRDH}
Let $n$ and $k$ be positive integers ($n>1$). We define the family RDH as follows:
\begin{align}\label{RDH for}
\text{RDH}:=\lbrace \Upsilon_{\mathbf{x}} \; : \; \mathbb{Z}_n^k \rightarrow \mathbb{Z}_n \; : \; \mathbf{x}\in {\mathbb{Z}_n^{*}}^k \rbrace,
\end{align}
where
\begin{align}\label{RDH for 2}
\Upsilon_{\mathbf{x}}(\mathbf{m}) := \mathbf{m} \cdot \mathbf{x} \pmod{n} = \sum_{i=1}^k m_ix_i \pmod{n},
\end{align}
for any $\mathbf{x}=\langle x_1, \ldots, x_k\rangle\in {\mathbb{Z}_n^{*}}^k$, and any $\mathbf{m}=\langle m_1, \ldots, m_k\rangle \in \mathbb{Z}_n^k$. Suppose that $t_1,\ldots,t_k$ are given positive divisors of $n$. Now, if in the definition of RDH instead of having $\mathbf{x}=\langle x_1, \ldots, x_k\rangle\in {\mathbb{Z}_n^{*}}^k$, we have, more generally, $\mathbf{x}=\langle x_1,\ldots,x_k \rangle \in \mathbb{Z}_n^k$ with $(x_i,n)=t_i$ ($1\leq i\leq k$), then we get a generalization of RDH that we call GRDH.
\end{definition}

It would be an interesting question to investigate for which values of $n$, GRDH is $\varepsilon$-AU or $\varepsilon$-A$\Delta$U. We now deal with these problems. The explicit formula for the number of solutions of restricted linear congruences (Theorem~\ref{th_gen_expl}) plays a key role here.

First, we prove the following lemma which is needed in proving the hashing results.

\begin{lemma}\label{lem: pop}
Let $k$ and $n$ be positive integers ($n>1$). For every prime divisor $p$ of $n$ let $r_p$ be the exponent of $p$ in the prime factorization of $n$. Also, suppose that $t_1,\ldots,t_k$ are given positive divisors of $n$. There are the following two cases:

\noindent\textit{(i)} If there exists some $i_0$ such that $t_{i_0}\not=1$ then there exists $\mathbf{a}=\langle a_1,\dots,a_k \rangle \in \mathbb{Z}_n^k \setminus \lbrace \mathbf{0} \rbrace$ such that for every prime $p\mid n$ we have $\mathfrak{m}_p(a_1,t_1,\ldots,a_k,t_k)>r_p$.

\noindent\textit{(ii)} If $t_i=1$ ($1\leq i\leq k$) then for every $\mathbf{a}=\langle a_1,\dots,a_k \rangle \in \mathbb{Z}_n^k \setminus \lbrace \mathbf{0} \rbrace$ there exists at least one 
prime $p\mid n$ such that $\mathfrak{m}_p(a_1,\ldots,a_k)\leq r_p$.
\end{lemma}

\begin{proof}
\noindent\textit{(i)} WLOG, let $t_1\not=1$, say, $t_1=t$ with $t\mid n$ and $t>1$. Take $a_1=\frac{n}{t}$ and $a_2=\cdots=a_k=0$. Now, for every prime $p\mid n$ we have $p^{r_p} \mid a_it_i$ ($1\leq i\leq k$). Therefore, for every prime $p\mid n$ we have $\mathfrak{m}_p(\frac{n}{t},t,0,t_2,\ldots,0,t_k)>r_p$.

\noindent\textit{(ii)} Let $t_i=1$ ($1\leq i\leq k$) and $\mathbf{a}=\langle a_1,\dots,a_k \rangle \in \mathbb{Z}_n^k \setminus \lbrace \mathbf{0} \rbrace$ be given. Suppose that for every prime $p\mid n$ we have $\mathfrak{m}_p(a_1,\ldots,a_k)> r_p$. This implies that for every prime $p\mid n$ we have $p^{r_p} \mid a_i$ ($1\leq i\leq k$). Therefore, we get $n \mid a_i$ ($1\leq i\leq k$) which is not possible because there exists some $i$ such that $a_i \in \mathbb{Z}_n \setminus \lbrace 0 \rbrace$.
\end{proof}

Now, we are ready to investigate the $\varepsilon$-almost-universality of GRDH.

\begin{theorem}\label{thm:GRDH e-au}
Let $n$ and $k$ be positive integers $(n,k>1)$. The family \textnormal{GRDH} is an $\varepsilon$-\textnormal{AU} family of hash functions for some $\varepsilon<1$ if and only if $n$ is odd and 
$(x_i,n)=t_i=1$ $(1\leq i\leq k)$. Furthermore, if these conditions are satisfied then \textnormal{GRDH} \textnormal{(}which is then reduced to \textnormal{RDH}\textnormal{)} is $\frac{1}{p-1}$-\textnormal{AU}, where $p$ is the smallest prime divisor of $n$. This bound is tight.
\end{theorem}

\begin{proof}
Assume the setting of the family GRDH, and that $\mathbf{t}=\langle t_1,\dots,t_k \rangle$ is given. Let $n>1$ and for every prime divisor $p$ of $n$ let $r_p$ be the exponent of $p$ in the prime factorization of $n$. Suppose that $\mathbf{m}=\langle m_1, \ldots, m_k \rangle \in \mathbb{Z}_n^k$ and $\mathbf{m}'= \langle m'_1, \ldots, m'_k \rangle \in \mathbb{Z}_n^k$ are any two distinct messages. Put $\mathbf{a}=\langle a_1,\dots,a_k \rangle = \mathbf{m}-\mathbf{m}'$. Since $\mathbf{m}\not=\mathbf{m}'$, there exists some $i$ such that $a_i \not= 0$. If in the family GRDH there is a collision between $\mathbf{m}$ and $\mathbf{m}'$, this means that there exists an $\mathbf{x}=\langle x_1,\ldots,x_k \rangle \in \mathbb{Z}_n^k$ with $(x_i,n)=t_i$, $t_i\mid n$ ($1\leq i\leq k$) such that $\Upsilon_{\mathbf{x}}(\mathbf{m})=\Upsilon_{\mathbf{x}}(\mathbf{m'})$. Clearly,
\begin{align*}
\Upsilon_{\mathbf{x}}(\mathbf{m})=\Upsilon_{\mathbf{x}}(\mathbf{m'}) \Longleftrightarrow \sum_{i=1}^k a_ix_i \equiv 0 \pmod{n}.
\end{align*}
So, we need to find the number of solutions $\mathbf{x}=\langle x_1,\ldots,x_k \rangle \in \mathbb{Z}_n^k$ of the restricted linear congruence $a_1x_1+\cdots +a_kx_k\equiv 0 \pmod{n}$, with $(x_i,n)=t_i$, $t_i\mid n$ ($1\leq i\leq k$). Here, since $b=0$, none of the two cases stated in the first part of Theorem~\ref{th_gen_expl} holds. Thus, by formula (\ref{main_prod_formula}), there are exactly
\begin{equation} \label{main_prod_formula2}
\mathlarger{\prod}_{i=1}^{k} \varphi\left(\frac{n}{t_i}\right)
\mathlarger{\prod}_{\substack{p\,\mid\, n \\ \mathfrak{m}_p \,\leq \, r_p}} p^{\mathfrak{m}_p-r_p-1} \left(1-\frac{(-1)^{e_p-1}}{(p-1)^{e_p-1}} \right)
\end{equation}
choices for such $\mathbf{x}=\langle x_1,\ldots,x_k \rangle \in \mathbb{Z}_n^k$ that satisfy the aforementioned restricted linear congruence, where the last product is over the prime factors $p$ of $n$ with $\mathfrak{m}_p \leq r_p$, $r_p$ is the exponent of $p$ in the prime factorization of $n$, 
$\mathfrak{m}_p$ is the smallest $j\geq 1$ such that there is some $i$ with $p^j \nmid a_it_i$, and 
$$
e_p = \# \{i: 1\leq i\leq k, p^{\mathfrak{m}_p}\nmid a_it_i \}.
$$
Also, since $(x_i,n)=t_i$ ($1\leq i\leq k$), the \textit{total} number of choices for $\langle x_1, \ldots, x_k \rangle$ is $\prod_{i=1}^{k} \varphi(\frac{n}{t_i})$. Therefore, given any $\mathbf{a}=\langle a_1,\dots,a_k \rangle \in \mathbb{Z}_n^k \setminus \lbrace \mathbf{0} \rbrace$, the collision probability is exactly 
\begin{align} \label{col prob given}
P_{\mathbf{a}}(n, \mathbf{t})=\mathlarger{\prod}_{\substack{p\,\mid\, n \\ \mathfrak{m}_p \,\leq \, r_p}} p^{\mathfrak{m}_p-r_p-1} \left(1-\frac{(-1)^{e_p-1}}{(p-1)^{e_p-1}} \right).
\end{align}

Now, there are two cases:

\noindent\textit{(i)} If for a prime $p \mid n$ we have $\mathfrak{m}_p \leq r_p$ then, by (\ref{col prob given}), the term corresponding to this $p$ in $P_{\mathbf{a}}(n, \mathbf{t})$ equals
\begin{align*}
p^{\mathfrak{m}_p-r_p-1} \left(1-\frac{(-1)^{e_p-1}}{(p-1)^{e_p-1}} \right)\leq p^{r_p-r_p-1}\left(1-\frac{(-1)^{2-1}}{(p-1)^{2-1}}\right)=\frac{1}{p-1}.
\end{align*}

\noindent\textit{(ii)} If for a prime $p \mid n$ we have $\mathfrak{m}_p > r_p$ then, by (\ref{col prob given}), the term corresponding to this $p$ in $P_{\mathbf{a}}(n, \mathbf{t})$ equals 1.

Let there exists some $i_0$ such that $t_{i_0}\not=1$. Then, by Lemma~\ref{lem: pop}(i), there exists $\mathbf{a}=\langle a_1,\dots,a_k \rangle \in \mathbb{Z}_n^k \setminus \lbrace \mathbf{0} \rbrace$ such that for every prime $p\mid n$ we have $\mathfrak{m}_p(a_1,t_1,\ldots,a_k,t_k)>r_p$. Now, by (\ref{col prob given}) and case (ii) above, the collision probability for this specific $\mathbf{a}$ is \textit{exactly one}. Now, assume that $t_i=1$ ($1\leq i\leq k$). Then, if $n$ is even, by taking $a_1=a_2=\frac{n}{2}$ and $a_3=\cdots=a_k=0$, one can see that 
$\mathfrak{m}_2(\frac{n}{2},\frac{n}{2},0,\ldots,0)=r_2$ and $e_2=2$, and for every other prime $p\mid n$ we have $\mathfrak{m}_p(\frac{n}{2},\frac{n}{2},0,\ldots,0)>r_p$. Now, by (\ref{col prob given}) and case (ii) above, the collision probability for this specific $\mathbf{a}$ is \textit{exactly one}.

Now, suppose that $n$ is odd and $t_i=1$ ($1\leq i\leq k$). Then, by Lemma~\ref{lem: pop}(ii), for every 
$\mathbf{a}=\langle a_1,\dots,a_k \rangle \in \mathbb{Z}_n^k \setminus \lbrace \mathbf{0} \rbrace$ there exists at least one prime $p\mid n$ such that $\mathfrak{m}_p(a_1,\ldots,a_k)\leq r_p$. Now, by (\ref{col prob given}) and cases (i), (ii) above, one can see that 
$$
\max_{\mathbf{a}=\mathbf{m}-\mathbf{m}' \in \mathbb{Z}_n^k \setminus \lbrace \mathbf{0} \rbrace} P_{\mathbf{a}}(n, \mathbf{t})
$$ 
is achieved in a specific $\mathbf{a}=\langle a_1,\dots,a_k \rangle \in \mathbb{Z}_n^k \setminus \lbrace \mathbf{0} \rbrace$ for which there exists \textit{exactly one} prime $p\mid n$ such that 
$\mathfrak{m}_p(a_1,\ldots,a_k)\leq r_p$, and furthermore, $p$ has to be the smallest prime divisor of 
$n$ that we denote by $p_{\min}$.

Consequently, if $n$ is odd and $(x_i,n)=t_i=1$ ($1\leq i\leq k$) then for any two distinct messages 
$\mathbf{m}, \mathbf{m}' \in \mathbb{Z}_n^k$, we have
\begin{align*}
\text{Pr}_{\Upsilon_{\mathbf{x}} \leftarrow \text{GRDH}}[\Upsilon_{\mathbf{x}}(\mathbf{m})=\Upsilon_{\mathbf{x}}(\mathbf{m'})] \leq \max_{\mathbf{a}=\mathbf{m}-\mathbf{m}' \in \mathbb{Z}_n^k \setminus \lbrace \mathbf{0} \rbrace} P_{\mathbf{a}}(n, \mathbf{t})\leq \frac{1}{p_{\min}-1}\leq \frac{1}{2}.
\end{align*}
Therefore, if $n$ is odd and $(x_i,n)=t_i=1$ ($1\leq i\leq k$) then GRDH (which is then reduced to RDH) is $\frac{1}{p_{\min}-1}$-\textnormal{AU}. We also note that this bound is tight: take $a_1=a_2=\frac{n}{p_{\min}}$ and $a_3=\cdots=a_k=0$. So, we get that $\mathfrak{m}_{p_{\min}}(\frac{n}{p_{\min}},\frac{n}{p_{\min}},0,\ldots,0)=r_{p_{\min}}$ and $e_{p_{\min}}=2$, and for every other prime $p\mid n$ we get that 
$\mathfrak{m}_p(\frac{n}{p_{\min}},\frac{n}{p_{\min}},0,\ldots,0)>r_p$. Now, by (\ref{col prob given}) and case (ii) above, the collision probability for this specific $\mathbf{a}$ is \textit{exactly} $\frac{1}{p_{\min}-1} \leq \frac{1}{2}$.
\end{proof}

The following remark gives a necessary and sufficient condition for the $\varepsilon$-almost-universality of the family GRDH in the case of $k=1$. We omit the proof as it is simply obtained from the above argument (this special case can be also proved directly, or, from \cite[Th. 3.1]{BKSTT}).

\begin{rema}\label{rem:GRDH one var}
If $k=1$ then the family \textnormal{GRDH} is an $\varepsilon$-\textnormal{AU} family of hash functions for some $\varepsilon<1$ if and only if $(x_1,n)=t_1=1$. Furthermore, if $(x_1,n)=t_1=1$ then the collision probability for any two distinct messages is `exactly zero'.
\end{rema}

Now, we investigate the $\varepsilon$-almost-$\Delta$-universality of GRDH. Note the change from $k>1$ in Theorem~\ref{thm:GRDH e-au} to $k\geq 1$ in Theorem~\ref{thm:GRDH e-adeltau}. The proof idea is similar to that of Theorem~\ref{thm:GRDH e-au}; so, in the proof we only write the parts which need more arguments.

\begin{theorem}\label{thm:GRDH e-adeltau}
Let $n$ and $k$ be positive integers $(n>1)$. The family \textnormal{GRDH} is an $\varepsilon$-\textnormal{A}$\Delta$\textnormal{U} family of hash functions for some $\varepsilon<1$ if and only if $n$ is odd and $(x_i,n)=t_i=1$ $(1\leq i\leq k)$. Furthermore, if these conditions are satisfied then \textnormal{GRDH} \textnormal{(}which is then reduced to \textnormal{RDH}\textnormal{)} is $\frac{1}{p-1}$-\textnormal{A}$\Delta$\textnormal{U}, where $p$ is the smallest prime divisor of $n$. This bound is tight.
\end{theorem}

\begin{proof}
Assume the setting of the family GRDH, and that $\mathbf{t}=\langle t_1,\dots,t_k \rangle$ is given. Let $n>1$ and for every prime divisor $p$ of $n$ let $r_p$ be the exponent of $p$ in the prime factorization of $n$. If for a given $\mathbf{a}=\langle a_1,\dots,a_k \rangle \in \mathbb{Z}_n^k \setminus \lbrace \mathbf{0} \rbrace$ and a given $b\in \mathbb{Z}_n$ there is a prime $p\mid n$ such that 
$\mathfrak{m}_p\leq r_p$ and $p^{\mathfrak{m}_p-1}\nmid b$, or, such that $\mathfrak{m}_p\geq r_p+1$ and $p^{r_p}\nmid b$, then, by the first part of Theorem~\ref{th_gen_expl}, the probability that we have 
$\Upsilon_{\mathbf{x}}(\mathbf{m})-\Upsilon_{\mathbf{x}}(\mathbf{m'})=b$ is \textit{exactly zero}. Otherwise, given any $\mathbf{a}=\langle a_1,\dots,a_k \rangle \in \mathbb{Z}_n^k \setminus \lbrace \mathbf{0} \rbrace$ and any $b\in \mathbb{Z}_n$, the probability that we have 
$\Upsilon_{\mathbf{x}}(\mathbf{m})-\Upsilon_{\mathbf{x}}(\mathbf{m'})=b$ is exactly 
\begin{align} \label{col prob given 2}
Q_{\mathbf{a}, b}(n, \mathbf{t})=\mathlarger{\prod}_{\substack{p\,\mid\, n \\ \mathfrak{m}_p \,\leq \, r_p \\ p^{\mathfrak{m}_p} \,\mid\, b}} p^{\mathfrak{m}_p-r_p-1} \left(1-\frac{(-1)^{e_p-1}}{(p-1)^{e_p-1}} \right)\mathlarger{\prod}_{\substack{p\, \mid\, n \\ \mathfrak{m}_p \,\leq \, r_p \\ p^{\mathfrak{m}_p-1} \, \|\, b}} p^{\mathfrak{m}_p-r_p-1} \left(1-\frac{(-1)^{e_p}}{(p-1)^{e_p}}\right).
\end{align}

Now, there are three cases:

\noindent\textit{(i)} If for a prime $p \mid n$ we have $\mathfrak{m}_p\leq r_p$ and $p^{\mathfrak{m}_p-1}\mid\mid b$ then, by (\ref{col prob given 2}), the term corresponding to this $p$ in $Q_{\mathbf{a}, b}(n, \mathbf{t})$ equals
\begin{align*}
p^{\mathfrak{m}_p-r_p-1}\left(1-\frac{(-1)^{e_p}}{(p-1)^{e_p}}\right)\leq p^{r_p-r_p-1}\left(1-\frac{(-1)^{1}}{(p-1)^{1}}\right)=\frac{1}{p-1}.
\end{align*}

\noindent\textit{(ii)} If for a prime $p \mid n$ we have $\mathfrak{m}_p \leq r_p$ and $p^{\mathfrak{m}_p} \mid b$ then, by (\ref{col prob given 2}), the term corresponding to this $p$ in $Q_{\mathbf{a}, b}(n, \mathbf{t})$ equals
\begin{align*}
p^{\mathfrak{m}_p-r_p-1}\left(1-\frac{(-1)^{e_p-1}}{(p-1)^{e_p-1}}\right)\leq p^{r_p-r_p-1}\left(1-\frac{(-1)^{2-1}}{(p-1)^{2-1}}\right)=\frac{1}{p-1}.
\end{align*} 

\noindent\textit{(iii)} If for a prime $p \mid n$ we have $\mathfrak{m}_p > r_p$ and $p^{r_p} \mid b$ then, by (\ref{col prob given 2}), the term corresponding to this $p$ in $Q_{\mathbf{a}, b}(n, \mathbf{t})$ equals 1.

If there exists some $i_0$ such that $t_{i_0}\not=1$ then the argument is exactly the same as before (just take $b=0$). Now, assume that $t_i=1$ ($1\leq i\leq k$). Then, if $n$ is even, take $a_1=b=\frac{n}{2}$ and $a_2=\cdots=a_k=0$. Now, one can see that, by (\ref{col prob given 2}) and case (iii) above, the probability that we have $\Upsilon_{\mathbf{x}}(\mathbf{m})-\Upsilon_{\mathbf{x}}(\mathbf{m'})=b$ for these specific $\mathbf{a}$ and $b$ is \textit{exactly one}.

Now, suppose that $n$ is odd and $t_i=1$ ($1\leq i\leq k$). Then, by (\ref{col prob given 2}), 
Lemma~\ref{lem: pop}(ii), and cases (i), (ii), (iii) above, one can see that 
$$
\max_{\substack{\mathbf{a}=\mathbf{m}-\mathbf{m}' \in \mathbb{Z}_n^k \setminus \lbrace \mathbf{0} \rbrace \\ b\in \mathbb{Z}_n}} Q_{\mathbf{a}, b}(n, \mathbf{t})
$$ 
is achieved in a specific $\mathbf{a}=\langle a_1,\dots,a_k \rangle \in \mathbb{Z}_n^k \setminus \lbrace \mathbf{0} \rbrace$ and a specific $b\in \mathbb{Z}_n$ for which there exists \textit{exactly one} 
prime $p\mid n$ such that $\mathfrak{m}_p(a_1,\ldots,a_k)\leq r_p$ and 
$p^{\mathfrak{m}_p-1}\mid\mid b$, or, $\mathfrak{m}_p(a_1,\ldots,a_k)\leq r_p$ and 
$p^{\mathfrak{m}_p} \mid b$, and also $p^{r_p} \mid b$ for every other prime $p\mid n$; furthermore, $p$ has to be the smallest prime divisor of $n$ that we denote by $p_{\min}$.

Consequently, if $n$ is odd and $(x_i,n)=t_i=1$ ($1\leq i\leq k$) then for any two distinct messages 
$\mathbf{m}, \mathbf{m}' \in \mathbb{Z}_n^k$, and all $b\in \mathbb{Z}_n$, we have
\begin{align*}
\text{Pr}_{\Upsilon_{\mathbf{x}} \leftarrow \text{GRDH}}[\Upsilon_{\mathbf{x}}(\mathbf{m})-\Upsilon_{\mathbf{x}}(\mathbf{m'})=b] \leq \max_{\substack{\mathbf{a}=\mathbf{m}-\mathbf{m}' \in \mathbb{Z}_n^k \setminus \lbrace \mathbf{0} \rbrace \\ b\in \mathbb{Z}_n}} Q_{\mathbf{a}, b}(n, \mathbf{t})\leq \frac{1}{p_{\min}-1}\leq \frac{1}{2}.
\end{align*}
Therefore, if $n$ is odd and $(x_i,n)=t_i=1$ ($1\leq i\leq k$) then GRDH (which is then reduced to RDH) is $\frac{1}{p_{\min}-1}$-A$\Delta$U. We also note that this bound is tight: take $a_1=b=\frac{n}{p_{\min}}$ and $a_2=\cdots=a_k=0$. Now, by (\ref{col prob given 2}) and case (iii) above, one can see that the probability that we have $\Upsilon_{\mathbf{x}}(\mathbf{m})-\Upsilon_{\mathbf{x}}(\mathbf{m'})=b$ for these specific $\mathbf{a}$ and $b$ is \textit{exactly} $\frac{1}{p_{\min}-1}$.
\end{proof}

\begin{rema}
While the proofs of Theorem~\ref{thm:GRDH e-au} and Theorem~\ref{thm:GRDH e-adeltau} are simple thanks to Theorem~\ref{th_gen_expl}, but there may be other simpler proofs (say, without relying on the counting arguments as we do) for these results. However, given the general statements of Theorem~\ref{thm:GRDH e-au} and Theorem~\ref{thm:GRDH e-adeltau}, possible simpler proofs for these results which cover the `whole' statements may not be necessarily that shorter. Besides, we believe that our proof techniques have their own merit and these connections and techniques may motivate more work in universal hashing and related areas.
\end{rema}

\begin{rema}
If in Theorem~\ref{thm:GRDH e-adeltau} we let $k=1$, then we get the main result of the paper by Alomair et al. \cite[Th. 5.11]{ACP} which was obtained via a very long argument.
\end{rema}

\begin{rema}
Using Theorem~\ref{th_gen_expl} and the idea of the proof of Theorem~\ref{thm:GRDH e-adeltau} one can see that there are cases in which the collision probability in the family \textnormal{GRDH} is `exactly zero' (Corollary~\ref{cor_zero} completely characterizes all these cases). This can be considered as an advantage of the family \textnormal{GRDH} and is not the case in the family \textnormal{MMH}$^*$, as the collision probability in \textnormal{MMH}$^*$ is always exactly $\frac{1}{p}$ which never vanishes.
\end{rema}

\section{Applications to authentication with secrecy}\label{Sec_4}

As an application of the results of the preceding section, we propose an authentication code with secrecy scheme which generalizes a recent construction \cite{ACP, AP}. We remark that Alomair et al. have applied their scheme in several other papers; see, e.g., \cite{ALP} for an application of this approach in the authentication problem in RFID systems. So, our results may have implications in those applications, as well. We adopt the notation of \cite{MSY} in specifying the syntax of these codes. In particular, we consider key-indexed families of coding rules.

An {\em authentication code with secrecy} (or {\em code} for short) is a tuple $\mathrm{C} = (\So,\Me,\K, \Enc,\Dec)$, specified by the following sets: $\So$  of {\em source states} (or {\em plaintexts}), $\Me$  of {\em messages} (or {\em ciphertexts}), $\K$ of {\em keys}, $\Enc$ of {\em authenticated encryption (AE) functions} and $\Dec$ of {\em verified decryption functions}. The sets $\Enc$ and $\Dec$ are indexed by 
$\K$. For $k \in \K$, $\Enc_k:\So\rightarrow\Me$ is the associated authenticated encryption function
and $D_k:\Me\rightarrow\So\cup\{\bot\}$ is the associated verified decryption function. The encryption and decryption functions have the property that for every $m \in \So$,  $\Dec_k(\Enc_k(m)) = m$. Moreover, for any $c \in \Me$, if $c \ne \Enc_k(m)$ for some $m \in \So$, $\Dec_k(c)=\bot$.

Before presenting our construction, we first note that although it is not explicitly stated in \cite{ACP, AP}, the construction given there is correct only for the case of a uniform distribution on source states. This will be the case for our construction, as well. We note that this assumption, while unrealistically strong from a security perspective, is commonly used in the study of authentication codes with secrecy.
Following the terminology of \cite{HubEq} (see also \cite{Hub2}), we will call such codes
{\em authentication and secrecy codes for equiprobable source probability distributions}. Henceforth we will work under the assumption of equiprobable source states.

We now give the security definitions required for authentication and secrecy. We begin with a definition of secrecy.
\begin{definition}
We say that $\mathrm{C}=(\So,\Me,\K,\Enc,\Dec)$ provides $\varepsilon$-{\em secrecy} on $\So'\subseteq \So$ if
every $m\in\So'$ and $c\in\Me$,
\[
\Pr_{m'\leftarrow\So,k\leftarrow\K}[m'=m|\Enc_k(m')=c]\le \varepsilon.
\]
Thus, $\frac{1}{|\So|}$-secrecy on $\So$ corresponds to the standard notion of Shannon secrecy \cite{SHA} 
(for a uniform message distribution).

With respect to authentication, we restrict attention to {\em substitution attacks}, also known as {\em spoofing attacks of order} 1. A $\mathrm{C}$-{\em forger} is a mapping $\mathcal{F}:\Me \rightarrow \Me$. Note that there are no computational restrictions on $\mathcal{F}$. We say that $\mathrm{C}$ is $\delta$-{\em secure against substitution attacks} if for every $\mathrm{C}$-forger $\mathcal{F}$,
\[
\Pr_{m\leftarrow\So,k\leftarrow\K,c\leftarrow\Enc_k(m)}[\mathcal{F}(c)\ne c \wedge \Dec_k(\mathcal{F}(c))\ne \bot] \le \delta.
\]
Finally, we say that $\mathrm{C}$ is  an $\varepsilon,\delta$-{\em authentication code with secrecy for equiprobable source states} on $\So'$ if it is $\varepsilon$-secret on $\So'$ and $\delta$-secure against substitution attacks.
\end{definition}

For any $n,k\in\mathbb{N}$, we define  $\mathrm{C}^{n,k}_{\mathrm{RDH}}$ as follows:  $\So=\Z_n^k$, $\K=\Z^k_n\times(\Z^*_n)^k$, $\Me=\Z_n^k\times \Z_n$. Thus, source states are $k$-tuples $\mathbf{m}=\langle m_1,\dots,m_k\rangle$, keys are pairs $\langle \mathbf{x},\mathbf{y}\rangle$ of $k$-tuples
$\mathbf{x}=\langle x_1,\dots,x_k\rangle$, $\mathbf{y}=\langle y_1,\dots,y_k\rangle$, and ciphertexts are pairs $\langle \mathbf{c},t\rangle$.

Note that we will sometimes write pairs using the notation $\cdot||\cdot$ rather than the usual $\langle \cdot,\cdot\rangle$, e.g., we write a key pair as $\mathbf{x}||\mathbf{y}$. Also, we may abuse terminology, and for a ciphertext $\mathbf{c}||t$, call $\mathbf{c}$ the ciphertext and $t$ the {\em tag}.
The authenticated encryption function $\Enc$ is defined as follows:
\[
\Enc_{\mathbf{x}||\mathbf{y}}(\mathbf{m})= \Psi_{\mathbf{x}}(\mathbf{m})||\Upsilon_{\mathbf{y}}(\mathbf{m}),
\]
where $\Upsilon$ is the RDH hash function, and
\[
\Psi_{\mathbf{x}}(\mathbf{m}) = \mathbf{m} + \mathbf{x} \pmod{n} = \langle m_1+x_1 \pmod{n}, \ldots, m_k+x_k \pmod{n}\rangle.
\]
To define $\Dec$, we first define $\Psi^{-1}$:
\[
\Psi^{-1}_{\mathbf{x}}(\mathbf{c}) = \mathbf{c} - \mathbf{x} \pmod{n} = \langle c_1-x_1 \pmod{n}, \ldots, c_k-x_k \pmod{n}\rangle.
\]
Then
\[
\Dec_{\mathbf{x}||\mathbf{y}}(\mathbf{c}||t)=
\left\{\begin{array}{ll}
\Psi^{-1}_{\mathbf{x}}(\mathbf{c}) & \text{if $\Upsilon_{\mathbf{y}}(\Psi^{-1}_{\mathbf{x}}(\mathbf{c}))=t$};\\
\bot & \text{otherwise}.
\end{array}\right.
\]

Now, we are ready to state and prove our main result in this section:

\begin{theorem} \label{authenc-main}
Let $n,k \in \mathbb{N}$, where $n$ is odd, and $p$ the smallest prime divisor of
$n$. Then $\mathrm{C}^{n,k}_{\mathrm{RDH}}$ is a $\frac{1}{(p-1)n^{k-1}},\frac{1}{p-1}$-authentication code with secrecy for equiprobable source states on $\Z_{n}^k \setminus \{\mathbf{0}\}$.
\end{theorem}

We will establish this theorem by the following sequence of lemmas.

\begin{lemma}
Let $n,k \in \mathbb{N}$, where $n$ is odd, and $p$ the smallest prime divisor of $n$. Then $\mathrm{C}^{n,k}_{\mathrm{RDH}}$ is $\frac{1}{(p-1)n^{k-1}}$-secret on $\Z_n^k\setminus \{\mathbf{0}\}$.
\end{lemma}

\begin{proof}
We first note that for any $\mathbf{m}$, $\mathbf{c}$, and $t$,
\[
\Pr_{\mathbf{m}',\mathbf{x}\leftarrow\Z_n^k,\mathbf{y}\leftarrow(\Z^*_n)^k}[\mathbf{m'}=\mathbf{m}|\Enc_{\mathbf{x}||\mathbf{y}}(\mathbf{m}')=\mathbf{c}||t]
=
\Pr_{\mathbf{m}'\leftarrow\Z_n^k,\mathbf{y}\leftarrow(\Z^*_n)^k}[\mathbf{m'}=\mathbf{m}|\Upsilon_{\mathbf{y}}(\mathbf{m}')=t].
\]
This follows from the independence of $\Psi_\mathbf{x}(\mathbf{m}')$ and $\Upsilon_\mathbf{y}(\mathbf{m'})$,  conditioned on $\mathbf{m}'=\mathbf{m}$, along with the fact that $\Psi$ provides Shannon secrecy. But
\begin{align*}
\Pr_{\mathbf{m}'\leftarrow\Z_n^k,\mathbf{y}\leftarrow(\Z^*_n)^k}[\mathbf{m'}=\mathbf{m}|\Upsilon_{\mathbf{y}}(\mathbf{m}')=t]
&=\Pr_{\mathbf{m}'\leftarrow\Z_n^k,\mathbf{y}\leftarrow(\Z^*_n)^k}[\Upsilon_{\mathbf{y}}(\mathbf{m}')=t|\mathbf{m'}=\mathbf{m}]/n^{k-1}\\
&\le \frac{1}{(p-1)n^{k-1}},
\end{align*}
where the equality follows by Bayes' rule and the fact that
for
$\mathbf{m'}\leftarrow(\Z_n)^k$ and $\mathbf{y}\leftarrow(\Z^*_n)^k$,
$\Upsilon_{\mathbf{y}}(\mathbf{m'})$ is uniformly distributed in $\Z_n$,
and the inequality follows, assuming $\mathbf{m}\ne\mathbf{0}$, by Theorem~\ref{thm:GRDH e-adeltau}. 
\end{proof}

We now establish a {\em key hiding} property which will be needed to prove resistance to substitution attacks.

\begin{lemma}
For $n,k \in \mathbb{N}$, $\mathbf{y}\in(\Z^*_n)^k$, $\mathbf{c}\in\Z_n^k$
and $t \in \Z_n$,
\[
\Pr_{\mathbf{x},\mathbf{m}\in\Z_n^k,\mathbf{y}'\in(\Z^*_n)^k}
[\mathbf{y}'=\mathbf{y}|\Enc_{\mathbf{x}||\mathbf{y}'}(\mathbf{m})=\mathbf{c}||t]=\frac{1}{|(\Z^*_n)^k|}.
\]
\end{lemma}

\begin{proof}
First note that since $\mathbf{x}$ and $\mathbf{m}$ are chosen independently of
$\mathbf{y}'$, it is the case that $\Psi_\mathbf{x}(\mathbf{m})$ and $\mathbf{y}'$ are independent.
So we just need to show that
\[ \Pr_{\mathbf{m}\in\Z_n^k,\mathbf{y}'\in(\Z^*_n)^k}
[\mathbf{y}'=\mathbf{y}|\Upsilon_{\mathbf{y}'}(\mathbf{m})=t] = \frac{1}{|(\Z^*_n)^k|}.
\]
Note that
\begin{align*}
\Pr_{\mathbf{m}\in\Z_n^k,\mathbf{y}'\in(\Z^*_n)^k}
[\Upsilon_{\mathbf{y}'}(\mathbf{m})=t|\mathbf{y}'=\mathbf{y}] &= \Pr_{\mathbf{m}\in\Z_n^k,\mathbf{y}'\in(\Z^*_n)^k}
[\Upsilon_{\mathbf{y}'}(\mathbf{m})=t \wedge \mathbf{y}'=\mathbf{y}]/
\Pr_{\mathbf{y}'\in(\Z^*_n)^k}
[\mathbf{y}'=\mathbf{y}]\\
&=\Pr_{\mathbf{m}\in\Z_n^k,\mathbf{y}'\in(\Z^*_n)^k}
[\Upsilon_{\mathbf{y}}(\mathbf{m})=t \wedge \mathbf{y}'=\mathbf{y}]/
\Pr_{\mathbf{y}'\in(\Z^*_n)^k}
[\mathbf{y}'=\mathbf{y}]\\
&=\Pr_{\mathbf{m}\in\Z_n^k}
[\Upsilon_{\mathbf{y}}(\mathbf{m})=t]\cdot\Pr_{\mathbf{y}'\in(\Z^*_n)^k} [\mathbf{y}'=\mathbf{y}]/
\Pr_{\mathbf{y}'\in(\Z^*_n)^k}
[\mathbf{y}'=\mathbf{y}]\\
&=\Pr_{\mathbf{m}\in\Z_n^k}
[\Upsilon_{\mathbf{y}}(\mathbf{m})=t]=\frac{1}{|\Z_n|},
\end{align*}
where the last equality follows because the product of a uniformly random element
of $\Z_n$ and a fixed element of $\Z^*_n$ is uniformly distributed in $\Z_n$, and
the sum of a fixed number of uniformly random elements of $\Z_n$ is uniformly
distributed in $\Z_n$.
We now have
\begin{multline}
\label{kh-Bayes}
\Pr_{\mathbf{m}\in\Z_n^k,\mathbf{y}'\in(\Z^*_n)^k}
[\mathbf{y}'=\mathbf{y}|\Upsilon_{\mathbf{y}'}(\mathbf{m})=t]\\ =
\Pr_{\mathbf{m}\in\Z_n^k,\mathbf{y}'\in(\Z^*_n)^k}
[\Upsilon_{\mathbf{y}'}(\mathbf{m})=t|\mathbf{y}'=\mathbf{y}]\cdot
\frac
{\Pr_{\mathbf{y}'\in(\Z^*_n)^k}[\mathbf{y}'=\mathbf{y}]}
{\Pr_{\mathbf{m}\in\Z_n^k,\mathbf{y}'\in(\Z^*_n)^k}[\Upsilon_{\mathbf{y}'}(\mathbf{m})=t]}.
\end{multline}
But
\begin{align*}
\Pr_{\mathbf{m}\in\Z_n^k,\mathbf{y}'\in(\Z^*_n)^k}[\Upsilon_{\mathbf{y}'}(\mathbf{m})=t]&=\sum_{\mathbf{y}\in(\Z^*_n)^k}\Pr_{\mathbf{m}\in\Z_n^k,\mathbf{y}'\in(\Z^*_n)^k}
[\Upsilon_{\mathbf{y}'}(\mathbf{m})=t|\mathbf{y}'=\mathbf{y}]\cdot\Pr_{\mathbf{y}'\in(\Z^*_n)^k}[\mathbf{y}'=\mathbf{y}]\\
&=\frac{1}{|\Z_n|}.
\end{align*}
Combining this with (\ref{kh-Bayes}) completes the proof.
\end{proof}
\begin{remark}
This key hiding property does not hold in general. The given proof depends on the fact that $\mathbf{m}$ is uniformly distributed in $\Z^k_n$.
\end{remark}
\begin{lemma}
Let $n,k \in \mathbb{N}$, where $n$ is odd, and $p$ the smallest prime divisor of $n$. Then $\mathrm{C}^{n,k}_{\mathrm{RDH}}$ is $\frac{1}{p-1}$-secure against substitution attacks.
\end{lemma}
\begin{proof} By way of contradiction suppose that $\mathcal{F}$ produces a substitution with probability greater than $\frac{1}{p-1}$. By averaging, there must be some $\mathbf{m}\in\Z_n^k$ such that if  $\Enc_{\mathbf{x}||\mathbf{y}}(\mathbf{m})=\mathbf{c}||t$, for random $\mathbf{x}$ and $\mathbf{y}$, then $\mathcal{F}(\mathbf{c}||t)=\mathbf{c}'||t'$ such that $\mathbf{c}'||t'\ne\mathbf{c}||t$ and $\Upsilon_{\mathbf{y}}(\Phi^{-1}_\mathbf{x})(\mathbf{c}')=t'$. Let $b=t-t'$ and $\mathbf{m}'=(\Phi^{-1}_\mathbf{x})(\mathbf{c}')$. Note that it must be the case that $\mathbf{m}'\ne\mathbf{m}$. By the preceding lemma, $\mathbf{y}$ and $\mathbf{m}'$ are statistically independent. So,
\[
\Upsilon_{\mathbf{y}}(\mathbf{m})
-\Upsilon_{\mathbf{y}}(\mathbf{m}')=b,
\]
for randomly chosen $\mathbf{y}\in(\Z^*_n)^k$, violating that RDH is $\frac{1}{p-1}$-A$\Delta$U by Theorem~\ref{thm:GRDH e-adeltau}.
\end{proof}

\subsection{Discussion}

The proposed scheme, which is a generalization of the scheme proposed in \cite{ACP, AP}, is defined using the {\em encrypt-and-authenticate} paradigm (see \cite{BN, KRA1} and the references therein, for a detailed discussion about these generic constructions and their security analysis). Since this approach requires the decryption of a purported ciphertext before its authentication, it is susceptible to attacks if the implementation of the decryption function leaks information when given invalid ciphertexts. Surprisingly, the preferred {\em encrypt-then-authenticate} approach will not work in our setting because it is not key-hiding.

We now show that the assumption that messages are generated uniformly at random is necessary for our result, by showing that any authentication scheme achieving $\varepsilon$-security against substitution attacks for arbitrary source distributions is in fact an $\varepsilon$-ASU hash family. 

We begin with some definitions.

\begin{definition}
A {\em authentication code} is specified by a tuple $\mathrm{M}=(\So,\Tag,\K,\Mac,\Auth)$ where  $\So$ is the set of {\em source states}, $\Tag$ is the set of {\em tags}, $\K$ is the set of {\em keys}, $\Mac:\K\times\So\rightarrow\Tag$, and $\Auth:\K\times\Tag\rightarrow\{0,1\}$. It must be the case that for all $k \in \K$ and $m \in \So$, $\Auth_k(m||\Mac_k(m))=1$. A {\em forger} is a mapping $\mathcal{F}=\langle\mathcal{F}_1,\mathcal{F}_2\rangle$ where $\mathcal{F}_1:\So\times\Tag\rightarrow\So$ and $\mathcal{F}_2:\So\times\Tag\rightarrow\Tag$. We say $\mathrm{M}$ is $\varepsilon$-{\em secure against substitution attacks} if for every forger $\mathcal{F}$ and distribution {\bf S} on $\So$,
\[
\Pr_{\substack{k\leftarrow\K,m\leftarrow_{\bf S}\So \\ t \leftarrow \Mac_k(m)}}[\mathcal{F}_1(m,t)\ne m \wedge \Auth_k(\mathcal{F}(m||t))=1]\le\varepsilon.
\]
\end{definition}

\begin{theorem}
Suppose $\mathrm{M}=(\So,\Tag,\K,\Mac,\Auth)$ is $\varepsilon$-secure against substitution attacks. Then 
$\{\Mac_k~|~k\in\K\}$ is an $\varepsilon$-\textnormal{ASU} hash function family.
\end{theorem}
\begin{proof}
Suppose $\{\Mac_k~|~k\in\K\}$ is not an $\varepsilon$-ASU hash family. So there are $m'\ne m'' \in \So$ and $t',t'' \in \Tag$ such that $\Pr_{k\leftarrow\K}[\Mac_k(m'')=t'' \wedge \Mac_k(m')=t'] > \varepsilon$.
Take  $\mathcal{F}$ such that $\mathcal{F}(m'||t')=m''||t''$, and let $\mathbf{S}$ be the distribution on $\So$ which puts all weight on $m'$. Then
\begin{align*}
\ & \Pr_{\substack{k\leftarrow\K,m\leftarrow_{\bf S}\So \\ t \leftarrow \Mac_k(m)}}[\mathcal{F}_1(m,t)\ne m \wedge \Auth_k(\mathcal{F}(m||t))=1]\\
=\ & \Pr_{\substack{k\leftarrow\K \\ t\leftarrow\Mac_k(m')}}[\mathcal{F}_1(m',t)\ne m'\wedge \Auth_k(\mathcal{F}(m'||t)=1]\\
=\ & \Pr_{\substack{k\leftarrow\K \\ t\leftarrow\Mac_k(m')}}[\mathcal{F}_1(m',t)\ne m'\wedge \Auth_k(\mathcal{F}(m'||t)=1|t=t']\cdot \Pr_{\substack{k\leftarrow\K \\ t\leftarrow\Mac_k(m')}}[t=t']\\
=\ & \Pr_{k\leftarrow\K}[\mathcal{F}_1(m',t')\ne m'\wedge \Auth_k(\mathcal{F}(m'||t')=1\wedge \Mac_k(m')=t']\\
=\ & \Pr_{k\leftarrow\K}[m''\ne m'\wedge \Mac_k(m'')=t'' \wedge \Mac_k(m')=t'] > \varepsilon.
\end{align*}
\end{proof}

\section*{Acknowledgements}

The authors would like to thank Martin Dietzfelbinger, Igor Shparlinski, and Roberto Tauraso for helpful comments on earlier versions of this paper. We are also grateful to the anonymous referees for helpful suggestions. During the preparation of this work the first author was supported by a Fellowship from the University of Victoria (UVic Fellowship).

\end{document}